\def\Anonym{0} 
\def\Eprint{1} 
\def\Draft{0} 
\def\Class{0} 
\def\IacrName{tosc}

\if\Class0
	\documentclass{llncs}
    \usepackage{geometry}
\fi
\if\Class1
	\documentclass[journal=\IacrName,preprint,spthm]{iacrtrans}
\fi

\usepackage[misc,geometry]{ifsym} 

\usepackage{graphicx}
\usepackage{url}
\usepackage{algorithm}
\usepackage{algorithmicx}
\usepackage[noend]{algpseudocode}
\usepackage[hidelinks,colorlinks,backref=page]{hyperref}

\usepackage{cite}
\usepackage{xifthen}
\usepackage{import}
\usepackage{bookmark} 

\usepackage{booktabs}
\usepackage{makecell}
\usepackage{multirow}
\usepackage{subcaption}
\usepackage{siunitx}
\usepackage{comment}

\usepackage{array}
\newcolumntype{L}[1]{>{\raggedright\let\newline\\\arraybackslash\hspace{0pt}}m{#1}}
\newcolumntype{C}[1]{>{\centering\let\newline\\\arraybackslash\hspace{0pt}}m{#1}}
\newcolumntype{R}[1]{>{\raggedleft\let\newline\\\arraybackslash\hspace{0pt}}m{#1}}

\if\Class0
\let\proof\relax

\fi

\usepackage{amsmath}
\usepackage{amsthm}
\usepackage{amssymb}
\usepackage{mathtools} 

\usepackage{mathrsfs} 
\usepackage{dsfont}
\usepackage{mathbbol} 

\newtheorem{notation}{Notation}
\newtheorem{observation}{Obsevation}

\usepackage{extarrows}  

\usepackage{xspace} 

\usepackage{tikz}
\usetikzlibrary{arrows}
\usetikzlibrary{patterns}
\usetikzlibrary{shapes}
\usetikzlibrary{calc}
\usetikzlibrary{decorations.pathreplacing}
\usepackage{combelow}

\algdef{SE}[SUBALG]{Indent}{EndIndent}{}{\algorithmicend\ }%
\algtext*{Indent}
\algtext*{EndIndent}

\usepackage{paralist} 
\usepackage{todo}

\usepackage[strings]{underscore} 
\graphicspath{{figures/}}

\newcommand\F{\mathbb{F}}

\newcommand\pround[1]{\left( #1 \right)}

\newcommand\pabs[1]{\left| #1 \right|}

\newcommand\ifPar[1]{%
	\ifthenelse{\isempty{#1}}%
		{}%
		{\pround{#1}}%
}
\newcommand\ifParElse[2]{
	\ifthenelse{\isempty{#1}}{#2}{(#1)}
}

\newcommand\eq[1]{\begin{align*}
#1
\end{align*}}

\algnewcommand{\LeftComment}[1]{\Statex \(\triangleright\) #1}

\newcommand\Input{\textbf{Input:}\xspace}
\newcommand\Output{\textbf{Output:}\xspace}

\title{
DenseQMC: an efficient bit-slice implementation of the Quine-McCluskey algorithm%
\if\Anonym0%
\texorpdfstring{\thanks{
\if\Eprint1
\fi
The work was supported by the Luxembourg National Research Fund's (FNR) and the German Research Foundation's (DFG) joint project APLICA (C19/IS/13641232).}
}{}
\fi
}

\if\Anonym0
\author{Aleksei Udovenko}

\institute{
	SnT, University of Luxembourg
	\\
	\email{aleksei@affine.group}
}
\fi

\if\Draft1
\newcommand\AU[1]{\Todo[AU]{\color{green!20!black}{#1}}}
\newcommand\AB[1]{\Todo[AB]{\color{blue!20!black}{#1}}}
\newcommand\AUdone[1]{\done\Todo[AU]{\color{green!20!black}{#1}}}
\newcommand\ABdone[1]{\done\Todo[AB]{\color{blue!20!black}{#1}}}
\else
\newcommand\AU[1]{}
\newcommand\AB[1]{}
\newcommand\AUdone[1]{}
\newcommand\ABdone[1]{}
\fi

\begin{document}
    \renewcommand{\sectionautorefname}{Section}
    \renewcommand{\subsectionautorefname}{Subsection}
    \if\Class0
    \newcommand{\algorithmautorefname}{Algorithm}
    \fi
    \newcommand\suppleref[1]{\hyperref[#1]{Supplementary Material~\ref{#1}}}

    \maketitle
    
    \newcommand\kws{
\keywords{
	Boolean minimization
	\and Two-level minimization
	\and Quine-McCluskey method
	\and Implementation
	\and Bit-slice
}
}
\if\Class1
\kws
\fi
\begin{abstract}
This note describes a new efficient bit-slice implementation DenseQMC of the Quine-McCluskey algorithm for finding all prime implicants of a Boolean function in the dense case. It is practically feasible for $n \le 23$ when run on a common laptop or for $n \le 27$ when run on a server with 1 TiB RAM.

This note also outlines a very common mistake in the implementations of the Quine-McCluskey algorithm, leading to a quadratic slowdown. An optimized corrected implementation of the classic approach is also given (called SparseQMC).

The implementation is freely available at \href{https://github.com/hellman/Quine-McCluskey}{github.com/hellman/Quine-McCluskey} .
\if\Class0
\kws
\fi
\end{abstract}

    \if\Draft1
    \setcounter{tocdepth}{10} 
    \fi
    
    \section{Introduction}

Boolean/logic minimization is a broad topic with numerous applications, including logic synthesis, expressing problems for SAT solvers (e.g., in symmetric-key cryptanalysis \cite{ToSC:SunWanWan21a}), performing Qualitative Comparative Analysis (QCA) in social sciences. The simplest minimization setting is the \emph{two-level minimization}, in which the target function has to be expressed in one of the two following types of formulas: \emph{sum-of-products} (SOP) or \emph{product-of-sums} (POS), also known as \emph{disjunctive normal form} (DNF) and \emph{conjunctive normal form} (CNF) respectively.
Here, a ``sum'' refers to the logic OR operator and a ``product'' refers to the logic AND operator.
The two formula types are dual to each other, as they are related by the De Morgan's laws. Therefore, it is sufficient to study any one of the two, and this work describes algorithms for DNF minimization.

A DNF/SOP formula consists of \emph{clauses} connected by the OR operator, and where each clause consists of variables or their negations connected by the AND operator. According to the sum/product interpretation, the OR operator is often denoted by ``+'' and the AND operator is omitted.
Furthermore, the logic negation (NOT) is denoted by the prime mark (for example, $x'$).
For example, the 3-bit majority function Maj can be expressed in DNF as \eq{
\text{Maj}(x,y,z) = xy + xz + yz.
}

The DNF minimization problem asks, given an expression of a Boolean function, to find a DNF formula with the smallest possible number of logic AND/OR operations (sometimes, only OR operations are counted, leading to minimization of the number of clauses). This work considers the most general case when the input function is given by its truth table.

\paragraph{The Quine-McCluskey method}
In a series of works \cite{quine1,quine2,quine3}, Quine and McCluskey developed an algorithm for two-level minimization. It consists of two steps:
\begin{enumerate}
\item Find all minimal products that can be consistently included in the DNF (also called \emph{prime implicants}). Minimality here means that no proper divisor of the product is consistent with the target function.
\item Choose a subset of minimal products according to the problem's goal: minimizing the number of products or minimizing the total number of AND/OR operations.
\end{enumerate}
The second step is an instance of the SetCover problem, which is NP-hard. Furthermore, NP-hard instances were proven to occur in this setup. Therefore, one has to resort to heuristic methods. Classically, the Petrick's method \cite{Petrick} was used. However, a modern integer optimization suite may perform better (as noted already in \cite{FCN61}). In addition, there are several dedicated exact or approximate heuristic solvers for the SetCover problem in general \cite{GYWL15,LMSZLL20,SZLLLM21}, as well as methods tailored to the Boolean minimization problem \cite{KS12}.

\paragraph{New method}
This work focuses on the first step of the Quine-McCluskey algorithm, namely, finding all prime implicants of the function. Although the second step is overwhelmingly dominating the full procedure, an efficient solution to the first step opens doors to fast heuristic approximate methods for the second step. In particular, the set of prime implicants itself provides a DNF expression of the function (or CNF of the function's complement) which can be sufficiently compact (even if containing redundant terms).

For the first step of the Quine-McCluskey algorithm, we propose a new implementation based on multidimensional ternary transforms and on the bit-slicing technique, called DenseQMC. To the best of our knowledge, the new approach significantly outperforms all available previous implementations and expands the feasible number $n$ of function's inputs to $n=23$ when run on a laptop and to $n=27$ when run on a server equipped with 1 TiB of RAM (see \autoref{tab:benchmark}). It is in particular useful for very dense problems (computing a DNF of a dense function or a CNF of a sparse function), which are the worst-case for the Quine-McCluskey algorithm and where sparse methods are too slow. For a fair comparison and due to absence of competitive classic (sparse) implementations, we also developed an optimized implementation of the classic Quine-McCluskey method, called SparseQMC. Surprisingly, it also outperforms existing reports and even the new DenseQMC when density is at most 50\%-60\%. 

A brief benchmark is given in \autoref{tab:benchmark}.
We do not provide explicit comparison with existing implementations, since we could not find any competitive implementation and/or sufficient performance information. For a rough comparison, the work \cite{Jain08} reports 34 seconds for a function with $n=11$, \cite{DusaThiem15} reports 5 seconds for a very sparse function with $n=15$, all of which are done instantly by any of our implementations; \cite{FPGA20} gives mixed CPU/GPU timings such as 1000 seconds and 10 GiB RAM for the dense case of $n=20$, 10 minutes for an $n=24$-bit function  of density 70\%, 2000 seconds for an $n=28$-bit function of density 30\%, close to $10^6$ seconds on an $n=32$-bit function of density 42\% using disk storage (due to ambiguous reports and absence of available implementation, it is difficult to provide a clear comparison). Note that the dense case is often occurring in practice when optimizing a CNF formula of a sparse function, for example, \cite{ToSC:BouCog20} report 2 hours of work for the case of $n=16$ and 82\%-dense function, appearing in cryptographic applications.
Note that this work does not intend to compete with sparse or approximate methods such as ESPRESSO-Exact \cite{VLSI84} or more recent ``Consistency Cubes'' method \cite{Dusa18}.

The source code is publicly available at
\begin{center}
\href{https://github.com/hellman/Quine-McCluskey}{github.com/hellman/Quine-McCluskey}
\end{center}

\begin{table}[htbp!]
	\centering
	\caption{Time/memory benchmark of the optimized implementations of new DenseQMC (\autoref{sec:dense},  \autoref{alg:bitslice}) and classic SparseQMC (\autoref{sec:sparse}) algorithms on random $n$-bit Boolean functions with a given density. Ran on a single core of an AMD EPYC 3.2 GHz processor with 1 TiB of RAM available.}
	\label{tab:benchmark}
	
	\vspace{0.25cm}

\setlength\tabcolsep{0.6em}
\begin{tabular}{c|rr|rr|rr|rr}
& \multicolumn{2}{c|}{\makecell{DenseQMC \\density: any}}
& \multicolumn{2}{c|}{\makecell{SparseQMC\\density: 25\%}}
& \multicolumn{2}{c|}{\makecell{SparseQMC\\density: 50\%}}
& \multicolumn{2}{c}{\makecell{SparseQMC\\density: 99\%}}
\\
\toprule
$n$
	& RAM & Time
	& RAM & Time
	& RAM & Time
	& RAM & Time
\\
\toprule
16
& 5 MiB & 0.01 \texttt{s}
& 2.0 MiB & 0.01 \texttt{s}
& 4.0 MiB & 0.07 \texttt{s}
& 0.2 GiB & 18 \texttt{s}
\\

17
& 16 MiB & 0.04 \texttt{s}
& 4.0 MiB & 0.03 \texttt{s}
& 16 MiB & 0.2 \texttt{s}
& 1.0 GiB & 62 \texttt{s}
\\

18
& 49 MiB & 0.2 \texttt{s}
& 8 MiB & 0.05 \texttt{s}
& 32 MiB & 0.7 \texttt{s}
& 2.0 GiB & 3.2 \texttt{m}
\\

19
& 0.1 GiB & 0.5 \texttt{s}
& 16 MiB & 0.2 \texttt{s}
& 64 MiB & 1.9 \texttt{s}
& 6 GiB & 11.2 \texttt{m}
\\

20
& 0.4 GiB & 1.5 \texttt{s}
& 32 MiB & 0.6 \texttt{s}
& 0.1 GiB & 5 \texttt{s}
& 16 GiB & 0.6 \texttt{h}
\\

21
& 1.3 GiB & 5 \texttt{s}
& 64 MiB & 1.6 \texttt{s}
& 0.2 GiB & 10 \texttt{s}
& 64 GiB & 1.6 \texttt{h}
\\

22
& 3.8 GiB & 13 \texttt{s}
& 0.1 GiB & 3 \texttt{s}
& 0.5 GiB & 22 \texttt{s}
& - & -
\\

23
& 12 GiB & 39 \texttt{s}
& 0.2 GiB & 7 \texttt{s}
& 1.5 GiB & 49 \texttt{s}
& - & -
\\

24
& 35 GiB & 2.0 \texttt{m}
& 0.5 GiB & 18 \texttt{s}
& 3.0 GiB & 1.9 \texttt{m}
& - & -
\\

25
& 0.1 TiB & 6.0 \texttt{m}
& 1.0 GiB & 35 \texttt{s}
& 6 GiB & 5.0 \texttt{m}
& - & -
\\

26
& 0.3 TiB & 0.3 \texttt{h}
& 2.0 GiB & 78 \texttt{s}
& 12 GiB & 13.1 \texttt{m}
& - & -
\\

27
& 0.9 TiB & 1.1 \texttt{h}
& 4.0 GiB & 3.2 \texttt{m}
& 24 GiB & 0.5 \texttt{h}
& - & -
\\

28
& - & -
& 8 GiB & 8.2 \texttt{m}
& 64 GiB & 1.2 \texttt{h}
& - & -
\\

29
& - & -
& 24 GiB & 0.3 \texttt{h}
& 0.1 TiB & 3.1 \texttt{h}
& - & -
\\

30
& - & -
& 48 GiB & 0.8 \texttt{h}
& - & -
& - & -
\\

31
& - & -
& 96 GiB & 1.8 \texttt{h}
& - & -
& - & -

\end{tabular}

\end{table}

\begin{figure}[htbp!]
	\centering
	\includegraphics[width=0.49\textwidth]{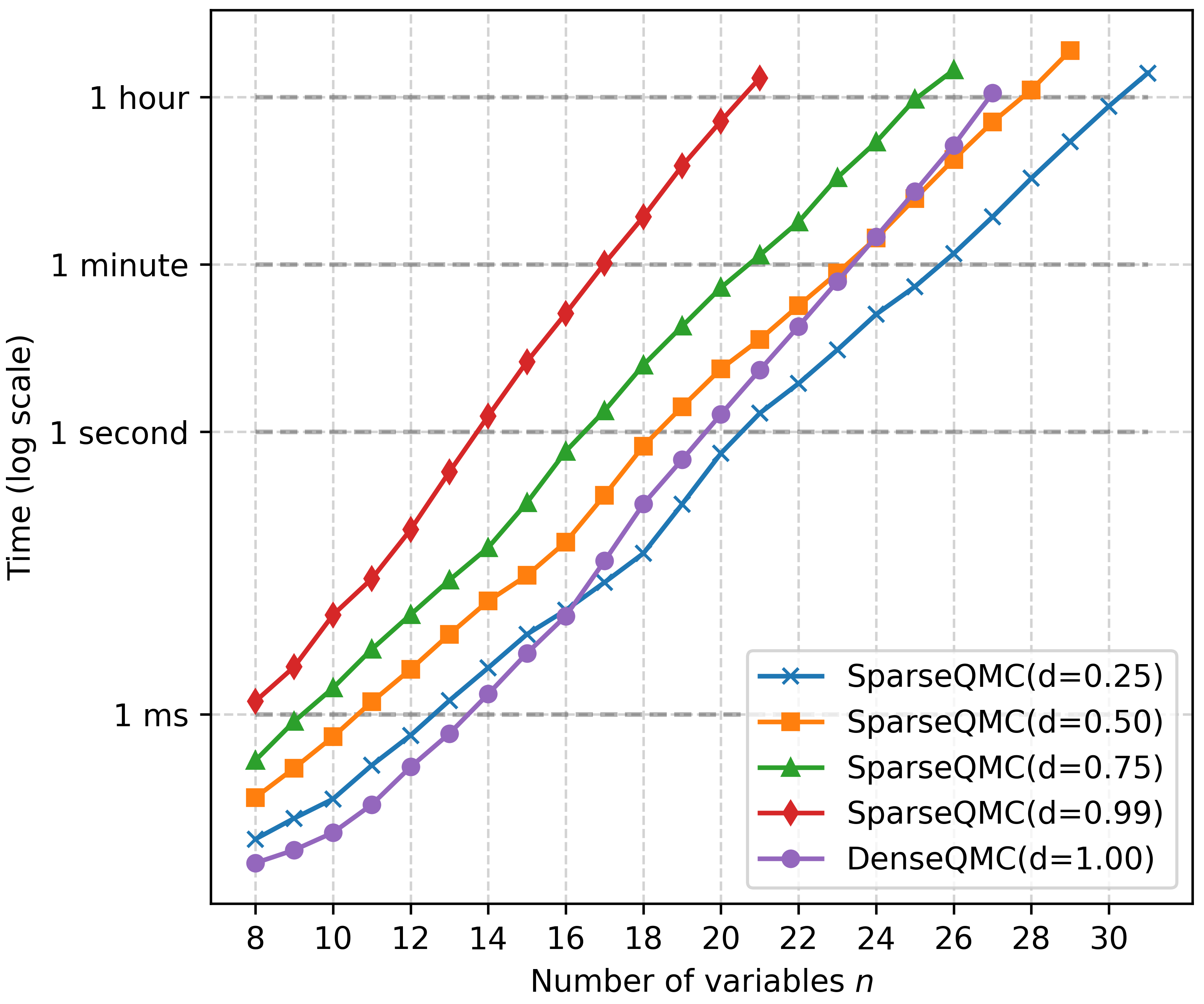}
	\includegraphics[width=0.49\textwidth]{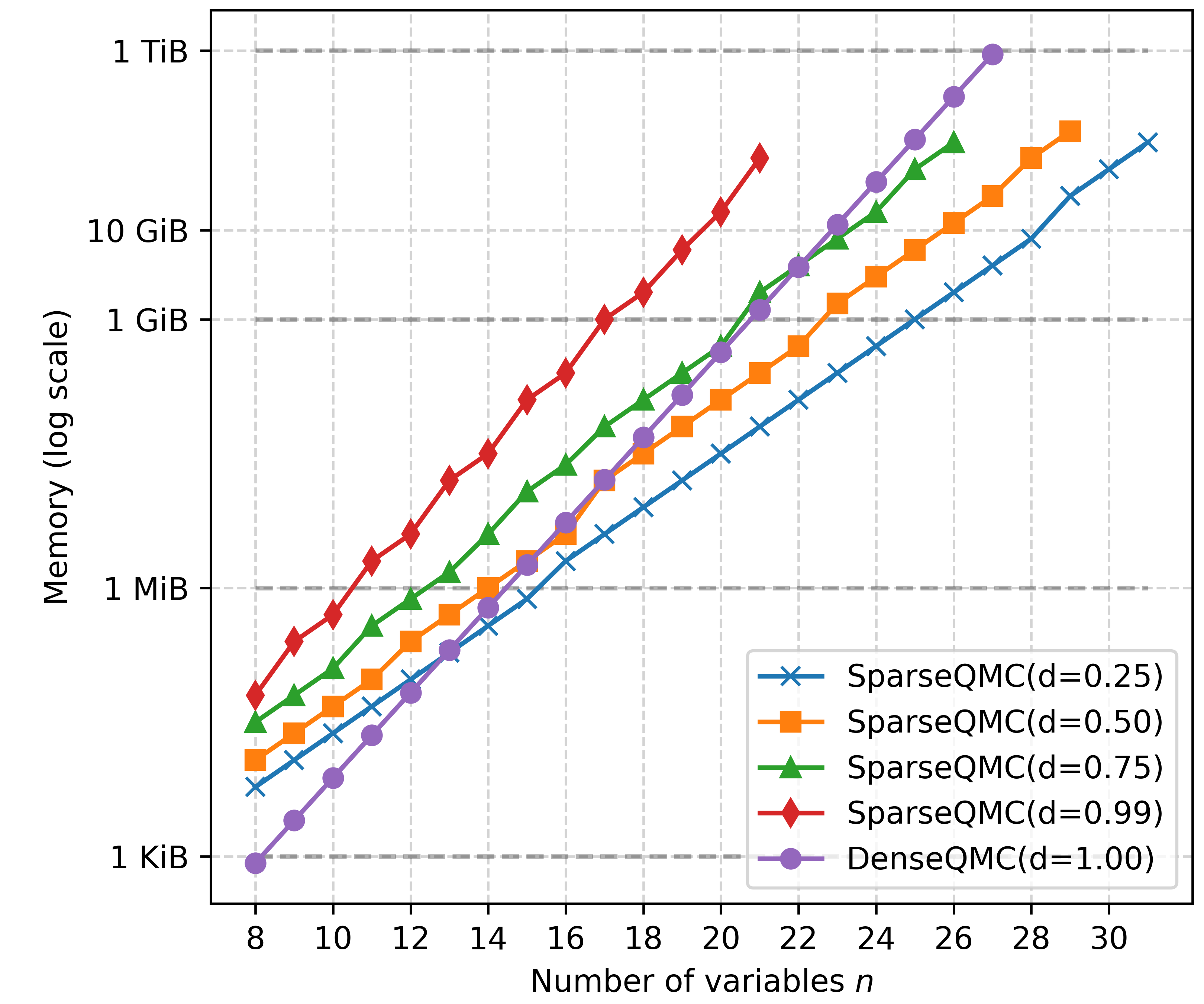}
	\caption{Comparison of DenseQMC and SparseQMC implementations (time and memory usage) on random $n$-bit Boolean functions with a given density $d$. Evaluated on a single core of a 3.2GHz CPU.}
	\label{fig:benchmark}
\end{figure}


\section{Definitions and notation}

\newcommand\zero{\mathtt{0}}
\newcommand\one{\mathtt{1}}
\renewcommand\star{\mathtt{*}}

The Boolean AND,OR,NOT operations are denoted by $\land,\lor,\lnot$ respective. They can be operate on single bits or bitwise on bit-vectors. The left and right shift operations on bit-vectors are denoted by $\lll$ and $\ggg$ respectively.

Let $\Sigma = \{\zero, \one, \star\}$ be the alphabet. The symbol $\star$ is called a \emph{wildcard}. Throughout the work, the number of input variables to the considered function is denoted by $n$.
A \emph{literal} $\alpha$ is either an input variable $x_i$ or its negation $\lnot x_i$.

\begin{definition}
	A \emph{minterm} is a product $\alpha_{1}\ldots \alpha_{m}$ of literals $\alpha_j \in \{x_{t_j}, x'_{t_j}\}$, $t_j \in \{1,\ldots,n\}$ such that $t_i\ne t_j$ for all $i \ne j$. In other words, every variable occurs at most once in the minterm. The minterm $\alpha_{1}\ldots \alpha_{m}$ is equivalently described by the string $s = (s_1,\ldots,s_n) \in \Sigma^n$ such that \eq{\begin{cases}
			s_k = \zero, &\text{if}~x'_{t_k}~\text{is present in the product},\\
			s_k = \one, &\text{if}~x_{t_k}~\text{is present in the product},\\
			s_k = \star, &\text{otherwise}.
	\end{cases}}
	The \emph{weight} of a minterm is defined as the number of wildcards in the string representation, equal to $n$ minus the minterm's degree.
\end{definition}


\begin{example}
	Let $n = 5$. Then, the minterm $x_1x_3'$ corresponds to the string $\one\star\zero\star\star$.
\end{example}

\begin{definition}
	A minterm $\alpha_{1}\ldots \alpha_{m}$ is called an \emph{implicant} of a Boolean function $f: \{0,1\}^n \to \{0,1\}$, if $\alpha_{1}\ldots \alpha_{m}(x) = 1$ implies $f(x) = 1$ for all $x \in \{0,1\}^n$.
\end{definition}

\begin{definition}
	An implicant $\alpha_{1}\ldots \alpha_{m}$ of a Boolean function $f$ is said to be \emph{prime}, if no proper divisor of $\alpha_{1}\ldots \alpha_{m}$ is an implicant of $f$. Otherwise, it is said to be \emph{redundant}.
\end{definition}


\section{Classic Quine-McCluskey algorithm}
\label{sec:classic}
The Quine-McCluskey algorithm constructs all implicants of the function, removing redundant implicants along the way, so that the final remaining implicants are prime. The construction relies on the following observation about combining minterms.

\begin{observation}
	Let $\alpha_1\ldots \alpha_m$ and $\beta_1\ldots \beta_m$ be two minterms such that, for some $i \in \{1,\ldots,m\}$, it is $\alpha_i = \beta_i'$ and $\alpha_j=\beta_j$ for all $j \ne i$. Then, the sum of $\alpha_1\ldots \alpha_m$ and $\beta_1\ldots \beta_m$ is functionally equivalent to the minterm $\prod_{j\ne i} a_j$.	
\end{observation}

In terms of strings, two minterms can be combined if they differ exactly at one position, and the new minterm is obtained by setting the string at this position to $\star$.
\begin{example}
The sum of minterms $\zero\zero\one\one\star$ and $\zero\one\one\one\star$ is $\zero\star\one\one\star$.
\end{example}

The observation works in the other way too: any implicant with at least one wildcard can be constructed as a sum of two implicants each with one wildcard less. This means that all implicants can be constructed bottom up from implicants with no wildcards.

The idea of the algorithm is to start with implicants of degree $n$ (having no wildcards in their strings). Then, at $w$-th level, $w \in \{1,\ldots,n\}$, all implicants with $w$ wildcards are constructed by combining implicants from the previous level in pairs. When a new implicant is constructed, the two used implicants are marked redundant (and can be removed from memory after the level $w$ is fully finished). The pseudocode for the high-level procedure is given in \autoref{alg:qmc}.

\begin{algorithm}
\caption{Quine-McCluskey algorithm for finding all prime implicants}
\label{alg:qmc}

\Input truth table of a Boolean $f: \{0,1\}^n\to\{0,1\}$

\Output set $S \subseteq \Sigma^n$ of prime implicants of $f$

\begin{algorithmic}[1]
\State $L_0 \gets \{x \in \{0,1\}^n : f(x) = 1\} \subseteq \Sigma^n$
\For{$w \in \{1,\ldots,n\}$}
	\State $R \gets \emptyset$
	\State $L_w \gets \emptyset$
	\For{$
		s, t \in L_{w-1}
		~:~
		\exists i \in \{1,\ldots,n\}
		~
		s_i=\zero,t_i=\one~\text{and}~s_j = t_j~\text{for all}~j\ne i
	$}
		\State $R \gets R \cup \{s,t\}$
		\State $L_w \gets L_w \cup \{s + t\}$
	\EndFor
	\State $L'_{w-1} \gets L'_{w-1} \setminus R$
\EndFor
\State $L'_n \gets L_n$
\State $S \gets \bigcup_{w=0}^n L'_w$
\end{algorithmic}
\end{algorithm}

\paragraph{Implementation caveat}
So far, we have not considered how to find pairs of implicants that differ in exactly one position (as strings). This crucial step is performed in line 5 of \autoref{alg:qmc}. Unfortunately, there is a widespread misconception about this step, stemming from a heuristic used in the initial works of Quine and McCluskey, who did not focus on formal algorithmic complexity of the approach.
In fact, most works from 19xx describe various heuristic variants in a ``pen-and-paper'' style.
Yet, even modern implementations follow the inefficient approach, see for example \cite{WWW03,Jain08,MQM12,SimCodes14,DusaThiem15}. A few works that follow the efficient approach are \cite{GPU17}. In the following, we first discuss the inefficient approach and then describe the efficient variants.

\paragraph{Finding compatible pairs of implicants (inefficient)}
Since the strings $s,t$ differ only at one position, at which they are equal to $\zero$ and $\one$ respectively, it follows that the number of $\one$s in $s$ is by one less than that in $t$. It is thus commonly suggested to sort and group each $L_{w-1}$ by the number of $\one$s in the increasing order. Let $L_{w-1,u}$ denote the subset of strings from $L_{w-1}$ which have exactly $u$ symbols $\one$, $0 \le u \le n-w+1$. Then, for each compatible pair $s,t$ considered by the algorithm, it must be $s \in L_{w-1,u}$ and $t \in L_{w-1,u+1}$ for some $u$. It follows that, at step $w$, it is sufficient to only consider pairs from the set \eq{
(L_{w-1,0}\times L_{w-1,1})~\cup~
\ldots
(L_{w-1,n-w}\times L_{w-1,n-w+1})
}
The problem is that this observation is often interpreted as the direct method to find the compatible pairs, by enumerating all pairs from $L_{w-1,u} \times L_{w-1,u+1}$, for each fitting $u$. This approach however leads to a quadratic complexity blowup.

\begin{proposition}
	There exists an infinite family of functions $f$ such that \autoref{alg:qmc} implemented using the described heuristic has \eq{
		\pabs{
			L_{w-1,u} \times L_{w-1,u+1}
		} = \Omega\pround{\frac{6^n}{n^3}}.
	}
	for some index $u$.
\end{proposition}
\begin{proof}
Consider the worst case for the Quine-McCluskey algorithm (in general), the constant-1 function $f=1$, for which all possible minterms are implicants (although only the minterm $1$ is prime).
For simplicity, we assume that $n=3m$ for some integer $m$.
Then, $L_{m,m}$ consists of strings from $\Sigma^n$ with exactly $m$ zeroes, $m$ ones, and $m$ wildcards, the number of which is 
\eq{
\frac{(3m)!}{m!m!m!} =
\Theta\pround{ 
\sqrt{m}\pround{\frac{3m}{e}}^{3m}
\cdot
\sqrt{m}^{-3}
\pround{\frac{m}{e}}^{-3m}
}
=
\Theta\pround{ 
\frac{3^{n}}{n}
}
}
for some constant $c$ (using Stirling's approximation). Since $L_{m,m}$ can be obtained by replacing one $\zero$ by $\one$ in each position in each string from $L_{m,m-1}$, we have $L_{m,m-1} \ge \frac{L_{m,m}}{m+1}$. It follows that for $w=m+1$, $u=m-1$, the set \eq{
	L_{w-1,u} \times L_{w-1,u+1}
}
has size $\Omega\pround{\frac{6^n}{n^3}}$.
\end{proof}

\paragraph{Finding compatible pairs of implicants (efficient)}

A correct implementation was described by Wegener \cite[Section 2.2]{Wegener87}, who suggested to iterate through the list $L_{w-1,u}$ and, for each string $s$ in it, check all strings that can be obtained by replacing exactly one $\one$ with $\zero$ inside the set $L_{w-1,u-1}$. Note that there are at most $n$ such strings to check. As a membership test, Wegener proposed to sort $L_{w-1,u-1}$ in advance and use binary search for each query. This leads to complexity $\mathcal{O}(3^nn^2)$ ($3^n$ maximum minterms, at most $n$ neighbors per a string, logarithmic cost of the binary search in $L_{w-1,u-1}$ of size at most $3^n$).
In practice, the membership test can be performed using a hash table in amortized constant time, leading to a probabilistic algorithm with expected time complexity of $\mathcal{O}(3^nn)$ operations on minterms ($n$-bit strings) and memory accesses.

Note that \emph{there is no benefit} for the asymptotic complexity in grouping elements of $L_{w-1}$ by the number of $\one$s: the cost of a membership test is asymptotically the same for sets of size $3^n$ and $3^n/n$ (both for binary search and hash-based methods). Although such grouping may improve locality of memory operations (and thus, improve the constant behind the asymptotic time complexity), it also makes the implementation more complex.


\section{DenseQMC: New bit-slice algorithm for the dense case}
\label{sec:dense}

This section describes a new algorithm for finding prime implicants, which exploits the ideas of Quine and McCluskey in a more careful implementation. It has complexity $\mathcal{O}(3^nn)$ of bit operations on \emph{any} function of $n$ variables, which is further reduced by a bit-slice style implementation allowing to maximally utilize the CPU registers.
Its high-level structure is based on iterating over \emph{dimensions} (variables), not over \emph{weights}. A similar idea was proposed already in \cite{Scheinman62}, but not sufficiently developed and detailed. Furthermore, we also employ and develop the idea of \cite{DusaThiem15} to consider the full possible state of size $3^n$ packed in a dense bit-vector similar to implementations in \cite{GPU17,FPGA20}. In addition, the core of the algorithm essentially describes a Boolean circuit, reminiscent of hardware accelerators proposed in \cite{FPGA19,FPGA20}.

\subsection{Multidimensional ternary transforms}

\paragraph{State representation}
The problem of finding prime implicants takes as input a subset of $\{0,1\}^n$ and outputs a subset of $\Sigma^n = \{\zero,\one,\star\}^n$. The input set can be also naturally embedded in $\Sigma^n$. This allows to represent the problem as a map from $2^{\Sigma^n}$ to itself. In the general (dense) case, it is efficient to represent a subset of $\Sigma^n$ as a bit-vector $S$ of size $3^n$ (i.e., $S \in \{0,1\}^{3^n}$). For now, we assume that bits in the memory are indexed by strings from $\Sigma^n$; concrete implementation is described in \autoref{sec:bit-slice}. For example, by $S(s)$ we denote the single bit indicating whether $s$ belongs to the set represented by $S$ or not. As we shall see, it is possible to implement the required mapping by manipulating such a state in-place by a fixed Boolean circuit.

The Quine-McCluskey algorithm can be directly reinterpreted to work on such a state. Essentially, we only change the data structure behind the sets $L_w$, which is a further evolution of the observation about efficient implementation in the previous section.

\newcommand\merge{\mathrm{Merge}}
\newcommand\reduce{\mathrm{Reduce}}
\newcommand\chizero{\chi_{\zero}}
\newcommand\chione{\chi_{\one}}
\newcommand\chistar{\chi_{\star}}

\begin{definition}
Define two Boolean operations \eq{
\merge &: \F_2^3 \to \F_2^3 : (\chizero,\chione,\chistar) ~\mapsto~
(\chizero,~ \chione,~ \chistar \lor (\chizero \land \chione)),\\
\reduce &: \F_2^3 \to \F_2^3 : (\chizero,\chione,\chistar)
~\mapsto~
(\chizero \land \lnot \chistar,~ \chione \land \lnot \chistar,~ \chistar).
}
\end{definition}
These operators can be viewed as maps from $2^{\Sigma}$ to itself. The $\merge$ operator will be used to combine two compatible minterms; the $\reduce$ operator will be used to remove redundant minterms.

\begin{theorem}
The Quine-McCluskey algorithm can be implemented as a Boolean circuit mapping the set $\{0,1\}^{3^n}$ to itself, interpreted as $2^{\Sigma^n}$. The circuit uses $3^{n-1}n$ Boolean OR and NOT gates, and $3^nn$ AND gates.
\end{theorem}
\begin{proof}
See \autoref{alg:qmc-bool}.
Complexity: the algorithm considers all triples having $n-1$ positions equal and the other position taking all three possible values (in the fixed order ($\zero,\one,\star$). There are $n3^{n-1}$ such triples. The $\merge$ and $\reduce$ operations are applied exactly once per such a triple.

Correctness: it is easy to see that $I$ represents the same state as the set $L_0 \cup L_1 \ldots \cup L_w$ from \autoref{alg:qmc} at matching steps. More precisely, each iteration of the first inner loop in \autoref{alg:qmc-bool} has the same effect as an iteration of the inner loop in \autoref{alg:qmc}; the second inner loop in \autoref{alg:qmc-bool} has the same effect as the line $L_w \gets L_w \setminus R$ in \autoref{alg:qmc}.
\end{proof}
	
\begin{algorithm}
	\caption{Boolean circuit implementation of the Quine-McCluskey algorithm for finding all prime implicants}
	\label{alg:qmc-bool}
	
	\Input indicator set $I \in \{0,1\}^{3^n}$ of the support  of a Boolean function $f: \{0,1\}^n \to \{0,1\}$
	
	\Output indicator set $O\in \{0,1\}^{3^n}$ of the set $O\subseteq \Sigma^n$ of prime implicants of $f$
	
	\begin{algorithmic}[1]
		\State $S \gets I$ \Comment{working state}
		\For{$w \in \{1,\ldots,n\}$}
		\For{$s,t,u \in \Sigma^n:$
			\Statex\hspace*{4em} 
				$u~\text{has}~w~\text{wildcards}$,
			\Statex\hspace*{4em}
				$\exists i\in\{1,\ldots,n\}~	s_i=\zero,t_i=\one,u_i=\star~\text{and}~s_j = t_j = u_j~\text{for all}~j\ne i$
		}
			\State $(S(s), S(t), S(v)) \gets \merge(S(s), S(t), S(v))$
		\EndFor
		\For{$s,t,u \in \Sigma^n:$
			\Statex\hspace*{4em} 
			$u~\text{has}~w~\text{wildcards}$,
			\Statex\hspace*{4em}
			$\exists i\in\{1,\ldots,n\}~	s_i=\zero,t_i=\one,u_i=\star~\text{and}~s_j = t_j = u_j~\text{for all}~j\ne i$
		}
			\State $(S(s), S(t), S(v)) \gets \reduce(S(s), S(t), S(v))$
		\EndFor
		\EndFor
		\State $O \gets S$
	\end{algorithmic}
\end{algorithm}

\begin{remark}
	Lines 2 and 4 of \autoref{alg:qmc-bool} define wirings in the Boolean circuit, and so do not contribute to its complexity. Software implementation of it would require careful enumeration of minterms with $w$ wildcards for all values of $w$. It is not included here as the alternative implementation below is superior.
\end{remark}


The Quine-McCluskey algorithm (and its Boolean circuit implementation from \autoref{alg:qmc-bool})
traverses the space of minterms in the order of increasing number of wildcards $w$ in the associated strings (equal to the decreasing order of minterms' degree). The correctness of the processing order is crucial, since $\merge$ transfers information from lower to higher $w$, while $\reduce$ goes in the other direction. A crucial observation is that these two processes can be completely separated. More precisely, second inner loop in \autoref{alg:qmc-bool} (Lines 4-5) can be performed in a separate loop on $w$ following the first loop on $w$. However, the second loop has to iterate $w$ in the decreasing order. In terms of implicants, the interpretation is that we can first compute all the implicants of the function, and only then remove all the redundant implicants.

\begin{lemma}
Let $\mathrm{MergeAll},\mathrm{ReduceAll}: \{0,1\}^{3^n} \to \{0,1\}^{3^n}$ be maps such that, for all $I \in \{0,1\}^{3^n}$, $M = \mathrm{MergeAll}(I)$, $O=\mathrm{ReduceAll}(M)$, it holds:
\eq{
	M(u) &=
	\begin{cases}
	I(u) &~~\text{if}~u \in \{\zero,\one\}^n,\\
	\bigvee_{s,t\in\Sigma^n:~s+t=u} M(s)\land M(t) &~~\text{otherwise},
	\end{cases}\\
	O(s) &= M(s) \land \lnot
		\bigwedge_{u,t \in \Sigma^n:~s+t=u}
		M(u).
}
Then:
\begin{enumerate}
\item For all $u \in \Sigma^n$ with $w$ wildcards,
$M(u)=1$
if and only if $u \in L_w$ computed in \autoref{alg:qmc}
if and only if $I(u) = 1$ at the iteration $w$ before the second inner loop in \autoref{alg:qmc-bool}.

\item For all $u \in \Sigma^n$ with $w$ wildcards,
$O(u)=1$
if and only if $u \in L'_w$ computed in \autoref{alg:qmc}
if and only if $I(u) = 1$  at the iteration $w$ after the second inner loop in \autoref{alg:qmc-bool}.
\end{enumerate}
\end{lemma}
\begin{proof}
Follows from the dataflow in the algorithms.
\end{proof}

%

We are now ready to explain the idea of multidimensional transforms, which essentially consists in switching the traversal order of the minterms from monotonic in the number of wildcards to traversal of each coordinate (``dimension'') at a time. In other words, the same operations \eq{
(S(s), S(t), S(v)) &\gets \merge(S(s), S(t), S(v)),\\
(S(s), S(t), S(v)) &\gets \reduce(S(s), S(t), S(v))
}
are performed, but on a differently ordered sequence of the triples $(s,t,v)$. The idea is to first iterate over the position $i$ where $s_i=\zero,t_i=\one,v_i=\star$. Note that all such triples are independent as any such pair of triples would differ in at least one one other position. All such triples can be enumerated for example by going through the whole string space $\Sigma^n$ filtered by $i$-th coordinate equal to $\zero$, and obtaining the respective $t,v$ by changing this coordinate appropriately.

The Boolean circuit utilizing the alternative traversal order is given by \autoref{alg:qmc-dim}.

\begin{algorithm}
\caption{Boolean circuit implementation of the Quine-McCluskey algorithm for finding all prime implicants (multidimensional transforms)}
\label{alg:qmc-dim}

\Input indicator set $I \in \{0,1\}^{3^n}$ of the support  of a Boolean function $f: \{0,1\}^n \to \{0,1\}$

\Output indicator set $O\in \{0,1\}^{3^n}$ of the set $O\subseteq \Sigma^n$ of prime implicants of $f$
	
\begin{algorithmic}[1]
\Statex $\mathrm{MergeAll}:$
\State $S \gets I$
\For{$i \in \{1, \ldots, n\}$}
\For{$s \in \Sigma^n: s_i = \zero$}
	\State $t \gets s$ with $i$-th position set to $\one$
	\State $u \gets s$ with $i$-th position set to $\star$
	\State $(S(s), S(t), S(v)) \gets \merge(S(s), S(t), S(v))$
\EndFor
\EndFor
\State $M \gets S$
\Statex $\mathrm{ReduceAll}:$
\State $S \gets M$
\For{$i \in \{1, \ldots, n\}$}
\For{$s \in \Sigma^n: s_i = \zero$}
\State $t \gets s$ with $i$-th position set to $\one$
\State $u \gets s$ with $i$-th position set to $\star$
\State $(S(s), S(t), S(v)) \gets \reduce(S(s), S(t), S(v))$
\EndFor
\EndFor
\State $O \gets S$
\end{algorithmic}
\end{algorithm}

\begin{proposition}
	\autoref{alg:qmc-bool} and \autoref{alg:qmc-dim} describe the same circuits (up to reordering gates).
\end{proposition}

Although the circuits produced by the two algorithms are equivalent, the new traversal order is more regular and easier to implement in software. More importantly, it allows efficient bit-slice implementation described below.

\subsection{Bit-slicing the transforms}
\label{sec:bit-slice}

While the Boolean circuit structure may be useful for theoretic purposes, for practical purposes it favors hardware implementations, which does not make sense for this kind of a problem. Therefore, it is crucial to optimize a software implementation of the algorithm. For this purpose, we employ the bit-slicing technique, which is (in particular) commonly used in implementations and designs of symmetric-key cryptographic primitives such as block ciphers \cite{FSE:GLSV14}.

The idea is that CPU instructions perform operations on full registers, not on single bits. The most common general register size is $\omega=64$ bits, while there are vector extensions processing even more bits at once. Essentially, a bitwise operation such as OR, AND, NOT can be performed on $\omega$ bits in parallel by a single instruction. A straightforward implementation of a Boolean circuit in software would waste this potential. Furthermore, a Boolean circuit needs to have a very regular structure to be efficiently bit-sliceable.
On the other hand, it is easy to perform \emph{batch} executions for any Boolean circuit: the $i$-th bit in each register is simply associated to the $i$-th parallel execution. Besides loading the inputs and reconstructing the outputs (which require single bit manipulations), the algorithm itself can be executed by simply translating the circuit into a straight-line program (using arbitrary topological order). In the following, we will use both the batch method and a custom bit-slice optimization relying on circuit's regular structure.

\paragraph{Data structure}
If the working state $S \in \{0,1\}^{3^N}$ is represented by a single bit-stream in the memory. 

We represent the working state $S \in \{0,1\}^{3^n}$ in a two-layer approach. In both layers, the strings over $\Sigma$ are treated as base-3 numbers (with $\star$ representing the digit 2). For both layers, we employ 0-based indexes, as they are more natural for algorithmic computations.

\begin{notation}
	For a string $s \in \Sigma^n$, we associate the integer $\rho(s) = \sum_{i=1}^n 3^{i-1} \rho(s_i)$, where $\rho(\zero)=0,\rho(\one)=1,\rho(\star) = 2$.
\end{notation}

The bottom layer covers $h$ dimensions, $1 \le h \le n$, while the top layer covers the remaining $n-h$ dimensions.
The whole working state consists of $3^{n-h}$ independent subsets of $\Sigma^{h}$, each represented by a contiguous bit-stream of length $3^{h}$ closely fitting a single CPU register or a small number of registers. All such bit-streams are always aligned to the register's least significant bit and padded to the register size. When such bit-streams are stored in memory, we will call them \emph{blocks}. Unused bits in blocks are not employed (i.e., are wasted). Therefore, it is desirable to choose a block size equal to a small multiple of the register size close to the maximum power of 3 it fits. Note that a large multiple implies bigger overhead on bitwise shifts.

\begin{example}
The 32-bit blocks are well fitted by $h=3$: 27/32 bits are used (15\% of used memory is wasted). The 64-bit blocks do not add any value over 32-bit blocks, while 128-bit blocks waste about 37\% of memory. The best fit on practice are 256-bit blocks: with $h=5$ only 5\% of memory is wasted, it is a small multiple of 64-bit register size provided by modern high-performance CPUs and also matches the bit-size of the AVX2 vector extension.
\end{example}

\begin{notation}
The block in a bottom layer will be denoted by $\mathrm{Block} \in \{0,1\}^{3^{h}}$. We recall that it represents a subset of $\Sigma^{3^{h}}$ by its indicator bit-string of length $3^{h}$, while requiring memory of $2^{\ell}$ bits, where $\ell$ is smallest integer such that $2^{\ell}\ge 3^{h}$.

A working state $S$ consists of $3^{n-h}$ blocks of size $3^{h}$ bits.
For $a \in \Sigma^{n-h}$ we let $S[a]$ denote the $\mathrm{Block}$ indexed by $\rho(a)$, and for $b \in \Sigma^{h}$ we let $S[a][b]$ denote the bit indexed by $\rho(b)$ in the block $S[a]$. It corresponds to the indicator bit for the minterm string $(a || b)$, where $||$ denotes the concatenation.
\end{notation}

\paragraph{Processing the top layer}
The top layer is straightforward to process as it can be seen as batch Boolean circuit evaluation, described above. To process $i$-th dimension, $1 \le i \le n-h$, we consider all triples of bits at distance $3^{i-1}$, non-intersecting and properly aligned.
This can be done by considering groups (of blocks) of size $3^i$. Inside each block at offset $b$, we consider first $3^{i-1}$ offsets $c$ inside the block, which correspond to minterms-strings having $\zero$ at position $i$. The respective strings having $\one$ and $\star$ at position $i$ are obtained by adding $3^{i-1}$ to the index. Thus, the 3 indexes corresponding to these strings are given by \eq{
\label{eq:indexes}
(i_s,i_t,i_u) \gets (b+c,~ b+c+3^{i-1},~ b+c+2\cdot3^{i-1}),
}
where $b$ takes values $0,3^i,2\cdot 3^i, \ldots, 3^{n-h}-3^i$; $c$ takes values $0,1,\ldots,3^{i-1}-1$.

\paragraph{Processing the bottom layer}
In the bottom layer, the same procedure is applied to each block independently. Recall that a block is essentially a register $v \in \{0,1\}^{3^h}$ storing the indicator vector of a subset of $\Sigma^{h}$ having $3^h$ bits, and extra unused bits which we can ignore. When processing a dimension $i$ covered by the bottom layer, we can extract the subset of relevant indicator bits by a computing AND with a fixed bitmask. For example, the $\zero$-valued positions in the lowest dimension  correspond to the mask
\eq{
	\mu_1 = (0, \ldots, 0,~ 0,0,1,~ 0,0,1,~ \ldots,~ 0, 0, 1)
}
(most-significant bits first), where the first group of zeroes corresponds to the padding (unused memory). The $\one$-valued positions can then be extracted by the same mask shifted by 1 position to the left:
\eq{
	(\mu_1 \lll 1) = (0, \ldots, 0,~ 0,1,0,~ 0,1,0,~ \ldots,~ 0, 1, 0),
}
and the $\star$-valued positions are given by one more shift:
\eq{
	(\mu_1 \lll 2) = (0, \ldots, 0,~ 1,0,0,~ 1,0,0,~ \ldots,~ 1, 0, 0),
}
Since the extracted bits need to be aligned to allow computations, it is more convenient to shift the register before applying the mask. Then, we can apply the $\merge/\reduce$ operations as:
\eq{
	s ~\gets~& (v \ggg 0) \land \mu_1,\\
	t ~\gets~& (v \ggg 1) \land \mu_1,\\
	u ~\gets~& (v \ggg 2) \land \mu_1,\\
(s,t,u) ~\gets~& \merge(s, t, u),\\
	v ~\gets~& (s \lll 0) \lor (t \lll 1) \lor (u \lll 2).
}
In this way, the operation is applied to $3^{h-1}$ bit triples at once. This is possible due to the regular structure of the algorithm, allowing aligning bits by simple shift operation.
This approach is easily generalized to any dimension $i$,
by computing the mask $\mu_i$ via summing bits at indexes given by $i_s$ from \eqref{alg:bitslice}, and replacing the shift amount by $3^{i-1}$.

\begin{notation}
	The $i$-th mask $\mu_i \in \{0,1\}^{3^h}$ is defined as  \eq{
		\mu_i = \sum_{b=0}^{3^h-i} \sum_{c=0}^{3^{i-1}} e_{3^ib + c},
	}
	where $e_j$ denotes the unit vector with $j$-th rightmost bit equal to 1.
\end{notation}

Note that, compared to the batch processing of the top layer, which almost fully exploits the available register (excluding only the unused ``padding'' bits), an operation in the bottom layer is performed on approximately only 1/3 of the register's bits in parallel. In addition, the shift and mask operations slow down the process even further. Yet, all this operations are performed locally on a single block, where as the top-layer processing requires more expensive memory accesses.

The full algorithm processing the two layers is given in \autoref{alg:bitslice}. The bit-slicing optimization can be summarized as full exploitation of a machine's word size, and generalized to the following theoretical result.

\begin{theorem}
In the RAM computational model with $\omega$-bit words, all prime implicants of an explicitly given Boolean set over $n$-bits can be computed in $\mathcal{O}(3^nn/\omega)$ bitwise word operations and memory accesses over a storage of $\mathcal{O}(3^n/\omega)$ words.
\end{theorem}

\begin{algorithm}
	\caption{Bit-slice implementation of the algorithm for finding all prime implicants}
	\label{alg:bitslice}
	
	\Input the support $I \subseteq \{0,1\}^{n}$ of a Boolean function $f: \{0,1\}^n \to \{0,1\}$
	
	\Output set $O\subseteq \Sigma^n$ of prime implicants of $f$
	(in the two-layer data structure)
	
	\begin{algorithmic}[1]
		\Statex Load input:
		\State $S \gets$ Two-layer data structure
		\For{$x \in I$}
			\State $(a, b) \in \Sigma^{n-h}\times \Sigma^{h} \gets x$ as a $\Sigma$-string
			\State $S[\rho(a)][\rho(b)] \gets 1$
		\EndFor
		\Statex
		\Statex $\mathrm{MergeAll}:$
		\Statex $\triangle$ Processing top layer (batch processing)
		\For{$i \in \{1, \ldots, n-h\}$}
		\For{$b \in 3^i\cdot \{0, \ldots, 3^{n-h-i}-1\}$}
		\For{$c \in \{0, \ldots, 3^{i-1}-1\}$}
			\State $(i_s,i_t,i_u) \gets (b+c,~ b+c+3^{i-1},~ b+c+2\cdot3^{i-1})$
			\State $(S[i_s], S[i_t], S[i_u]) \gets \merge(S[i_s], S[i_t], S[i_u])$
				\Comment{bit-slice operation on blocks}
		\EndFor
		\EndFor
		\EndFor
		\Statex $\triangle$ Processing bottom layer (bit-masks manipulation on single blocks)
		\For{$a \in \{0,\ldots,3^{n-h}-1\}$} \Comment{iterate over all blocks}
			\For{$i \in \{1,\ldots,h\}$}
			\State $s \gets (S[a] \ggg 0\cdot 3^{i-1}) \land \mu_i$
			\State $t \gets (S[a] \ggg 1\cdot 3^{i-1}) \land \mu_i$
			\State $u \gets (S[a] \ggg 2\cdot 3^{i-1}) \land \mu_i$
			\State $(s,t,u) \gets \merge(s, t, u)$
			\State $S[a] \gets (s \lll 0\cdot3^{i-1}) \lor (t \lll 1\cdot 3^{i-1}) \lor (u \lll 2\cdot3^{i-1})$
		\EndFor
		\EndFor
		%
		%
		\Statex
		\Statex $\mathrm{ReduceAll}:$
		\State ~~~~Same operations as in the $\mathrm{MergeAll}$ procedure above,
		\Statex ~~~~but with $\reduce$ instead of $\merge$.
		
		\Statex
		\Statex Extracting output:
		\State $O \gets \emptyset$
		\For{$a \in \Sigma^{n-h},~ b \in \Sigma^h$}
			\If{$S[\rho(a)][\rho(b)] = 1$}
				\State $O \gets O \cup (a||b)$
			\EndIf
		\EndFor
	\end{algorithmic}
\end{algorithm}


\section{SparseQMC: Optimized hash-based implementation for the sparse case}
\label{sec:sparse}

For a fair comparison, we implemented the classic Quine-McCluskey algorithm using hash-based neighbor search and various optimizations, which we describe in this section. This implementation is particularly useful in the sparse case, including the case of random functions with density of 50\%.

There also exist a few other heuristic algorithms for the sparse case \cite{SCLL70,Sen83}, which however are not sufficiently detailed to be directly and efficiently implemented. For example, these methods often involve manipulation of formulas defining the function per each prime implicant

\paragraph{Bit-based representation of minterms}
Each minterm is represented by a $2n$-bit value (fitting a 64-bit machine word when $n \le 31$, assuming 2 bits are reserved for control flags). Each 2-bit chunk represents a symbol from $\Sigma$ by correspondences 
$\mathtt{00}_2 \mapsto \zero$,
$\mathtt{01}_2 \mapsto \one$,
$\mathtt{10}_2 \mapsto \star$.
This allows efficient minterm manipulation and hashing by bitwise operations. A similar idea was proposed in \cite{Scheinman62,Jain08}.

\paragraph{Avoiding repeated implicant generation}
Each implicant containing $w$ wildcards can be constructed in $w$ possible pairs of implicants with $w-1$ wildcards (per each wildcard position). To avoid duplicates, it is sufficient to merge minterms on positions that don't have preceding wildcards. For example, $\one\star\star$ can be merged from $(\one\zero\star,\one\one\star)$ and from $(\one\star\zero,\one\star\one)$. The latter one can be skipped as there is a wildcard preceding the merging position 3. However, such skipped pairs might not be marked as redundant as they should. Therefore, an extra pass over all new implicants is required to mark all $w$ possible merge pairs as redundant.

This optimization allows, at step $w$, to first collect all new implicants in a list (which is more efficient than a hash-table due to sequential memory access), and then to convert it to the hash-based set $L_w$ (requiring 1 random memory access per element if linear probing is used, see below). Then, all $w$ merge pairs of each new implicant has to be marked redundant (removed from $L_{w-1}$).

\paragraph{Custom hash-based set}
As described above, the step $w$ of \autoref{alg:qmc} can be processed in two separate stages: generating the full set of implicants $L_w$, and removing redundant implicants from $L_{w-1}$. A hash-table with linear probing reduces the amount of random memory accesses, which is crucial for large memory structures. While generally linear probing becomes complicated when elements are removed, the two-stage process simplifies the implementation: when elements are removed, they can be simply marked in the hash-space (e.g., by using a control bit) without performing the expensive shrinking procedure requiring to rehash the whole set. At the end of the step, non-redundant (prime) implicants are collected and the set can be cleared to free the memory.


\section{Experimental evaluation}

We implemented the new dense algorithm and the optimized sparse variant in the C++ programming language. In addition, for the dense algorithm, we added obvious optimizations, such as unrolling the bottom layer dimension loop and merging the bottom layers of $\merge$ and $\reduce$ (using the fact that the two $\reduce$ layers are independent and can be swapped). The optimized version performs about 3 times faster than the basic one (measured on $n=22$).


Experiments were performed on a single core of an AMD EPYC 3.2 GHz CPU, run on the Ubuntu Server OS \footnote{\url{https://ubuntu.com/download/server}} inside a QEMU \footnote{\url{https://www.qemu.org/}} virtual machine, with 1 TiB of RAM available.
The code was compiled using the GNU \texttt{g++} compiler (version 9.4.0).
Time measurements and memory requirements are listed in \autoref{tab:benchmark} and illustrated in \autoref{fig:benchmark}. Measurements were done on single runs due to time constraints.

The source code is publicly available at
\begin{center}
	\href{https://github.com/hellman/Quine-McCluskey}{github.com/hellman/Quine-McCluskey}
\end{center}


%



	\if\Class0
	    \bibliographystyle{splncs04}
	\else
	    \bibliographystyle{alpha}
	\fi
    \bibliography{bib/abbrev3,bib/custom,bib/crypto}

    \appendix

%
    
    \todos
\end{document}